\documentclass[dvips,12pt]{article}
\usepackage{amsmath}
\usepackage{amsfonts}
\usepackage{amsthm}
\usepackage{amssymb}
\usepackage{amstext}
\usepackage{amsopn}
\usepackage{mathrsfs}
\usepackage{subfigure}
\usepackage{graphicx}
\usepackage[left=2.3cm,top=2cm,right=2.3cm,bottom=2cm, nohead,foot=1cm]{geometry}
\numberwithin{equation}{section}

\newtheorem{theorem}{Theorem}[section]
\newtheorem{lemma}[theorem]{Lemma}

\newtheorem{proposition}[theorem]{Proposition}
\newtheorem{definition}[theorem]{Definition}

\newcommand{\R}{{\mathbb R}}

\newcommand{\Hi}{{\mathcal{H}}}
\newcommand{\Bi}{{B}}

\newcommand{\Tr}{{\textup{Tr}}}

\newcommand{\C}{{\mathbb C}}

\newcommand{\scat}{{\mathbf{\mathit{S}}}}
\newcommand{\Sca}{{\mathbf{S}}}
\newcommand{\Asca}{{\mathbf{A}}}

\newcommand{\ato}{{ \vec{a}_{\vec{k},\lambda}(\vec{P})  }}
\newcommand{\att}{{ \vec{a}_{\vec{k},\lambda}(\vec{P})   }}

\newcommand{\yto}{{ \vec{a}_{\vec{k},\sigma,\lambda}(\vec{P})+(\sigma-I)\vec{k}  }}

\newcommand{\kto}{{ \vec{d}_{\vec{k},\lambda}(\vec{P})  }}
\newcommand{\ktt}{{ \vec{d}_{\vec{k},\lambda}(\vec{P})   }}

\newcommand{\vt}{{ \vec{v}_{\vec{k},\sigma,\lambda}(\vec{P})  }}

\newcommand{\vto}{{ \vec{v}_{\vec{k},\sigma_{1},\lambda}(\vec{P})  }}
\newcommand{\vtt}{{ \vec{v}_{\vec{k},\sigma_{2},\lambda}(\vec{P})   }}

\newcommand{\wt}{{\vec{w}_{\vec{k},\sigma,\lambda}(\vec{P})  }}

\newcommand{\wto}{{\vec{w}_{\vec{k},\sigma_{1},\lambda}(\vec{P})  }}
\newcommand{\wtt}{{ \vec{w}_{\vec{k},\sigma_{2},\lambda}(\vec{P})   }}

\newcommand{\mtoj}{{m_{j,\vec{k},\sigma_{1},\lambda}(\vec{P})  }}
\newcommand{\mttj}{{ m_{j,\vec{k},\sigma_{2},\lambda}(\vec{P})   }}

\title{An infinite-temperature limit for a quantum scattering process  }
\author{\textbf{Jeremy Clark}\\ jeremy.clark@fys.kuleuven.be \\ Katholieke Universiteit Leuven, Instituut voor Theoretische Fysica\\  Celestijnenlaan 200d, 3001 Heverlee, Belgium }

\begin{document}
\maketitle

\begin{abstract}
We study a quantum dynamical semigroup driven by a Lindblad generator with a deterministic Schr\"odinger part and  a noisy Poission-timed scattering part.  The dynamics describes the evolution of a test particle in $\R^{n}$, $n=1,2,3$, immersed in a gas, and the noisy scattering part is defined by the reduced effect of an individual interaction,  where the interaction between the test particle and a single gas particle is via a repulsive point potential.    In the limit that the mass ratio $\lambda=\frac{m}{M}$ tends to zero  and the collisions become more frequent as $\frac{1}{\lambda}$, we show that our dynamics $\Phi_{t,\lambda}$ approaches a limiting dynamics $\Phi_{t,\lambda}^{\diamond}$ with second order error.    Working in the Heisenberg representation, for $G\in \Bi(L^{2}(\R^{n}))$ $n=1,3$ we bound the difference between $\Phi_{t,\lambda}(G)$ and $\Phi_{t,\lambda}^{\diamond}(G)$ in operator norm proportional to $\lambda^{2}$. 

\end{abstract}

\section{Introduction}

The study of spatial decoherence in atom optics has inspired the derivation of certain Markovian master equations as models~\cite{Vacchini} for the reduced dynamics of particle interacting with an environment.    A Markovian approximation for the particle is made possible in part by the assumption that the degrees of freedom of the environment operate on a much shorter time scale than the particle.   In this case, the individual interactions between the particle and the environment are effectively instantaneous with respect to the time scale of the particle.      Thus many of the derivations of decoherence models in atom optics begin with an analysis of the scattering operator for the interaction between the particle and a single member of the reservoir~\cite{Joos,ESL,Sipe}.  A study of models where this scattering assumption can be made more rigorous can be found in~\cite{Two,Carlone,AsymSysm}.   

The current article concerns a quantum Markovian dynamics simulating the evolution of a large test particle of mass $M$ immersed in an inhomogeneous gas of light, high-speed particles of mass $m$.  We begin with a dynamics semigroup $\Phi_{t,\lambda}$ in the Heisenberg picture operating on $\Bi(L^{2}(\R^{n}) )$  and governed by an equation of the form:
\begin{align}\label{PoissonProc}
\frac{d}{dt}\Phi_{t,\lambda}(G)=i\big[ \frac{\vec{P} ^{2}}{2M}, \Phi_{t,\lambda}(G)\big]+\frac{1}{\lambda}\big(\Tr_{2}[ (I\otimes \rho) \Sca^{*}_{\lambda}(\Phi_{t,\lambda}(G)\otimes I)\Sca_{\lambda} ]-\Phi_{t,\lambda}(G)\big),
\end{align}
where $\Phi_{0,\lambda}(G)=G $, $\lambda=\frac{m}{M}$,  $\vec{P}$ is the vector of momentum operators, $\rho$ is a positive trace class operator on the Hilbert space of a single reservoir particle $L^{2}(\R^{n})$, and $\Sca_{\lambda}\in \Bi(L^{2}(\R^{n})\otimes L^{2}(\R^{n}))$ is the unitary scattering operator for a repulsive point interaction.      The scattering operator  $\Sca_{\lambda}$ for two particles which interact through a point interaction is determined through center-of-mass coordinates by the scattering operator for a single particle in a point potential.    Expressions for the scattering operator for a particle in a point potential can be found in~\cite{Solvable}.    Non-trivial point potentials only exist in dimensions $1$, $2$, and $3$.   Our parametrization of $\Sca_{\lambda}$ by $\lambda$ corresponds to  holding the strength of the interaction fixed while varying the mass ratio and a discussion can be found in~\cite{Clark}.

Previously it has been shown in~\cite{Clark} that, under some norm conditions on $\rho$,  there exists a $c>0$ such that for all $\lambda\geq 0$ and $G\in \Bi(L^{2}(\R^{n}))$
\begin{multline}\label{Previous}
\| \frac{1}{\lambda}\big(\Tr_{2}[ (I\otimes \rho) \Sca^{*}_{\lambda}(G\otimes I)\Sca_{\lambda} ]-G\big)\\- \big( i[V_{1}+\lambda V_{2}+\frac{\lambda}{2}\{\vec{P},  \vec{A}\} , G ]+\lambda (\varphi(G)-\frac{1}{2}\{ \varphi(I),G\}) \big) \|\leq c\lambda^{2}\|G\|_{wn},
\end{multline}
where $V_{1},V_{2},(\vec{A})_{j}$ are bounded real functions of the vector $\vec{X}$ of position operators, $\varphi$ is a bounded completely positive map of the form:
\begin{align}\label{Mult}
\varphi(G)=\sum_{j} \int_{\R^{n}}d\vec{k}\, m_{j,\vec{k}}^{*}\,G\, m_{j,\vec{k}},
\end{align}
where $m_{j,\vec{k} }$ are functions of $\vec{X}$, and $\|\cdot \|_{wn}$ is a weighted operator norm   of the form
$$\|G\|_{wn}=\|G\|+\| |\vec{X}|G\|+\|G|\vec{X}|\|\\ +\sum_{0\leq i,j \leq d}(\|X_{i}P_{j}G\|+\|GP_{j}X_{i}\|)+\sum_{ e_{1}+e_{2}\leq 3}\| |\vec{P}|^{e_{1}}G |\vec{P}|^{e_{2}}\|.$$     
The forms for $V_{1},V_{2},\vec{A}$ and $\varphi $ can be found in Appendix~\ref{Form} and depend on $\rho$.

In the current article, the main result is to extend~(\ref{Previous}) to an inequality bounding the difference between the semigroup $\Phi_{t,\lambda}$ and a limiting semigroup $\Phi_{t,\lambda}^{\diamond}$:
\begin{align}\label{MainEqu}
\|\Phi_{t,\lambda}(G)-\Phi_{t,\lambda}^{\diamond}(G)\|_{wn}\leq \lambda^{2}(t\wedge 1) \kappa_{T}\|G\|_{wn},     
\end{align}
for some fixed constant $\kappa_{T}$, for all $\lambda>0$, $t$ in a time interval $[0,T]$, and $G\in \Bi(L^{2}(\R^{n}))$.   $\Phi_{t,\lambda}^{\diamond}$ is a semigroup satisfying an equation of the form:
\begin{multline}\label{LimitEQN}
\frac{d}{dt}\Phi_{t,\lambda}^{\diamond}(G)= i[\frac{1}{2M}(\vec{P}+\lambda M \vec{A})^{2}+V_{1}+\lambda V_{2} , \Phi_{t,\lambda}^{\diamond}(G) ]\\ +\lambda \big(\varphi(\Phi_{t,\lambda}^{\diamond}(G))-\frac{1}{2}\{\Phi_{t,\lambda}^{\diamond}(G),\varphi(I)\}            \big).
\end{multline}
Notice that we have added the second order term $\frac{i}{2}\lambda^{2}M[\vec{A}^{2},\cdot]$ to the generator to complete the square with $\frac{i}{2M}[\vec{P}^{2},\cdot ]$.   The dynamics $\Phi_{t,\lambda}$ can be constructed through a Dyson series since $\varphi$ is bounded and the left term of the generator written in~(\ref{LimitEQN}) describes an electromagnetic field and is guaranteed to generate a well defined  group of isometries even for certain classes of unbounded $V_{1}$, $V_{2}$, and $\vec{A}$~\cite{EM}.   The equations~(\ref{PoissonProc}) and~(\ref{LimitEQN}) formally have the Lindblad form~\cite{Lindblad} except that the Hamiltonian part of the generator is unbounded in both cases due to the presence of  $\vec{P}$.  

The equations~(\ref{PoissonProc}) and~(\ref{LimitEQN}) should be regarded as formal, since the generators involve unbounded operators (namely the momentum operators $P_{j}$).    Lindblad equations with unbounded generators can be made rigorous using form-generator maps $\mathcal{L}(\psi_{1};G; \psi_{2})\rightarrow \C$ where $\psi_{1},\psi_{2}$ are in a dense domain $D\subset L^{2}(\R^{n})$ and $G\in \mathcal{B}(L^{2}(\R^{n}))$.   A solution $\Phi_{t}:\Bi(L^{2}(\R^{n}))$ should satisfy
$$\frac{d}{dt}\langle \psi_{1}| \Phi_{t}(G) \psi_{2} \rangle= \mathcal{L}(\psi_{1};\Phi_{t}(G);\psi_{2}),$$ 
for $\Phi_{0}(G)=G\in \Bi(\L^{2}(\R^{n}))$ and all $\psi_{1},\psi_{2}\in D$. 
The form $\mathcal{L}(\psi_{1};G;\psi_{2})$ in antilinear in $\psi_{1}$ and linear in $G$ and $\psi_{2}$, and its action is always clear from a formal expression for the unbounded Lindblad generator.    Form generators are discussed in~\cite{Cheb,Holevo2} and  references therein.   It is shown  in Lemma~\ref{CPSG}, that a rigorous interpretation of~(\ref{PoissonProc}) and~(\ref{LimitEQN})  yields the existence of unique semigroups $\Phi_{t,\lambda}$, $\Phi_{t,\lambda}^{\diamond}$ of completely positive maps satisfying $\Phi_{t,\lambda}(I)=\Phi_{t,\lambda}^{\diamond}(I)=I$.    For an introductory discussion of quantum dynamical semigroups see~\cite{Dyn}. 

The dynamics $\Phi_{t,\lambda}$ can be interpreted as the reduced dynamics, in the Heisenberg picture, for a test particle immersed in a spatially inhomogeneous gas where the effective state of the gas particles seen by the test particle is $\rho$.   The use of the scattering matrix $\Sca_{\lambda}$ in the noise term of~(\ref{PoissonProc}) assumes that the gas is sparse enough so that non-trivial interactions between the test particle and gas particles occur individually and that interactions occur on a faster time scale than the kinetic motion of the test particle.  If the duration of individual interactions is negligible compared to the time scale of the kinetic evolution of the test particle, the relevant information of an interaction is the translation of the ``before picture" to the ``after picture" determined by the scattering operator. 

The limiting regime $\lambda=\frac{m}{M}\ll 1$ corresponds to a reservoir of low mass, high speed gas particles.  By expanding our noise term in~(\ref{PoissonProc}) around $\lambda=\frac{m}{M}=0$ with the test particle mass $M$ and the effective density matrix of a single gas particle $\rho$ held fixed, the fixed momentum distribution $\rho(\vec{p},\vec{p})$ determines growing velocities as $\vec{v}(\vec{p})=\frac{\vec{p}}{m}=\frac{\vec{p}}{\lambda\,M}$.  For the same reason, small $\lambda$ corresponds to a high temperature regime: $E(\vec{p})=\frac{\vec{p}^{2}}{2m}=\frac{\vec{p}^{2}}{2\lambda M}$.  The noise term in~(\ref{PoissonProc}) has $\lambda$ dependence in the interaction frequency coefficient $\frac{1}{\lambda}$ which grows on the order of the reservoir particle speeds and $\lambda$ dependence in the scattering operators $\Sca_{\lambda}$.  In the limit $\lambda\rightarrow 0$, $\Sca_{\lambda}$ converges weakly to the identity operator.  This has the interpretation of quantum tunneling, since a fast reservoir particle will tend to tunnel over the potential of the test particle.  Without the frequency coefficient $\frac{1}{\lambda}$, the noise term would then tend to zero as $\lambda\rightarrow 0$.  Intuitively, a well-defined limiting dynamics emerges since the increased rate of tunneling through the potential of the test particle is compensated by the increased frequency of opportunities for collisions with the high speed gas particles.  

By expanding $\lambda$ in the noise term around the infinite temperature regime $\lambda=0$ only to first order, our limiting model neglects a description of energy relaxation.  This disadvantage is shared by the models~\cite{Joos,ESL} which are derived in similar limits involving expressions for scattering and small mass ratio $\frac{m}{M}\ll 1$.  Since the limit $\lambda\rightarrow 0$ involves the particles in the reservoir moving comparatively infinitely fast compared to the test particle, the first order approximation yields an environmental noise whose action is invariant of the momentum of the test particle.  Clearly, an environmental noise term which is invariant of test particle's momentum can not describe the relaxation of the particle's state to the temperature of the reservoir.  To see the invariance mathematically, let $W_{(\vec{q},\vec{p})}=e^{i\vec{q}\vec{P}+i\vec{p}\vec{X}}$ be the Weyl operator corresponding to a space shift by $\vec{q}$ and a momentum boost by $\vec{p}$, then the noise part of the generator in~(\ref{LimitEQN}) satisfies the momentum boost covariance:
\begin{align}\label{Trans}
\varphi(W_{(0,\vec{p})}^{*}\,G\,W_{(0,\vec{p})})-\frac{1}{2}\{ \varphi(I), W_{(0,\vec{p})}^{*}\,G\,W_{(0,\vec{p})}\} =   W_{(0,\vec{p})}^{*}\big(\varphi(G)-\frac{1}{2}\{\varphi(I), G \}\big)   W_{(0,\vec{p})}.
\end{align}
We conjecture that energy relaxation will be described by the dynamics generated by a second-order expansion in $\lambda$ of the noise term in~(\ref{PoissonProc}).  Existence of dynamical semigroups with unbounded Linblad generators satisfying momentum boost covariance~(\ref{LimitEQN}) is briefly discussed in~\cite{Cov}.   

Since $\rho$ is a trace class operator, it is Hilbert Schmidt, and thus the integral kernel $\rho(\vec{x}_{1},\vec{x}_{2})$ in the position basis is in $L^{2}(\R^{n}\times \R^{n})$.   Hence $\rho(\vec{x}_{1},\vec{x}_{2})$ is small (at least in Hilbert-Schmidt norm) outside of some compact set.   We interpret the dynamics $\Phi_{t,\lambda}$ as valid for a finite time period in which, intuitively, the  test particle is contained in a region where the effective state $\rho$ is relevant.  In the limiting dynamics $\Phi_{t,\lambda}^{\diamond}$, $V_{1}+\lambda V_{2}$ acts as a field potential and  $\lambda M\vec{A}$ acts as a vector potential. The presence of the $\vec{A}$, $V_{1}+\lambda V_{2}$ terms and also the form of the map $\varphi$ for the limiting dynamics $\Phi_{t,\lambda}^{\diamond}$ differs from those yielded by a related limit (involving a scattering process for a massive particle in a gas) found in~\cite{ESL,Sipe}, and the difference is due to the spatial inhomogeneity of the gas described by $\rho$.   In their setup, the limiting evolution of an observable (i.e. the Heisenberg representation) $G_{t}$ satisfies a differential equation of the form
\begin{align}\label{Before}
 \frac{d}{dt}G_{t}= i\big[\frac{\vec{P}^{2}}{2M}, G_{t}\big]+\int_{\R^{3}}d\vec{k}\, r(\vec{k})\big(e^{i\vec{X}\vec{k} }G_{t} e^{-i\vec{X}\vec{k} }-G_{t}\big), 
\end{align}
for a specified Poisson density rate $r(\vec{k})$ determining the rate of momentum boosts $\vec{k}$ transferred to the test particle.   

  We can make a connection between our limiting dynamics~(\ref{LimitEQN}) and an equation of the form~(\ref{Before}) by considering a parametrized collection of positive trace-class operators $\rho_{\epsilon}$ which will describe a spatially uniform environment in the limit $\epsilon\rightarrow 0$. Define $\rho_{\epsilon}$ through integral kernels in the momentum representation as
$$\rho_{\epsilon}(\vec{k}_{1},\vec{k}_{2})= \epsilon^{-(n-1)} \int_{\R^{n}} d\vec{k}\,P(\vec{k})\,\frac{e^{-\frac{1}{8\epsilon}(\vec{k}_{1}+\vec{k}_{2}-2\vec{k})^{2}-\frac{1}{8\epsilon}(\vec{k}_{1}-\vec{k}_{2})^{2}}  }{\sqrt{2\pi \epsilon}}    $$
for some probability density $P(\vec{k})$ on.   As $\epsilon\rightarrow 0$, $\rho_{\epsilon}(\vec{k}_{1},\vec{k}_{2})$ converges to zero for $\vec{k}_{1}\neq \vec{k}_{2}$ and asymptotically has the order of $\epsilon^{-(n-1)}P(\vec{k})$  for $\vec{k}=\vec{k}_{1}=\vec{k}_{2}$.  Using the expressions in Appendix~\ref{Form}, we see that $V_{1}(\epsilon)$, $V_{2}(\epsilon)$, and $\vec{A}(\epsilon)$ converge to constants (and thus can be gauged out), and $\varphi_{\epsilon}(G) $ converges to an expression as in~(\ref{Before}) which is proportional to 
\begin{align}\label{Ancient}
\int_{\R^{n}} d\vec{k}\, P(\vec{k})\,\int_{|\vec{k}|=|\vec{v}|}d\vec{v}\,e^{i\vec{X}(-\vec{v}+\vec{k})}G e^{-i\vec{X}(-\vec{v}+\vec{k})}, 
\end{align}
so that $\varphi_{0}(G)-\frac{1}{2}\{\varphi_{0}(I),G\}$ has the form of the noise in~(\ref{Before}).  The natural interpretation is that an incoming reservoir particle of momentum $\vec{k}$ in momentum transfers of $-\vec{v}+\vec{k}$ with $\vec{v}$ on the shell of radius $|\vec{k}|$.   The Kraus decomposition suggested by~(\ref{Ancient}) is an integral combination of conjugation by unitary Weyl operators corresponding to momentum shifts.      
 This contrasts with the general situation for $\varphi$ in which the terms $m_{j,\vec{k}}$ determining the Kraus decomposition~(\ref{General}) are non-unitary combinations of Weyl operators  $e^{i\vec{X}(-\vec{v}+\vec{k})}$ for fixed $\vec{k}$ and varying $\vec{v}$ on the surface $|\vec{v}|=|\vec{k}|$.

   The dynamics $\Phi_{t,\lambda}$ can be constructed as
\begin{align}\label{PoissonInt}
\Phi_{t,\lambda}(G)=e^{-\frac{t}{\lambda} }\sum_{n=0}^{\infty}\frac{1}{\lambda^{n} }\int_{0\leq t_{1}\cdots t_{n} \leq t} F_{t-t_{n}}\mathcal{P}_{\lambda} \cdots  \mathcal{P}_{\lambda}F_{t_{1}}(G),
\end{align}
where $\mathcal{P}_{\lambda}(G)=\Tr_{2}[(I\otimes \rho)\Sca_{\lambda}^{*}(G\otimes I) \Sca_{\lambda}]$ and $F_{t}$ is generated by $i[\frac{\vec{P}^{2}}{2M},\cdot]$.     Since $\mathcal{P}_{\lambda}$ and $F_{t}$ are completely positive maps, compositions and sums of them will be also, and hence the $\Phi_{t,\lambda}$'s are completely positive.   The construction~(\ref{PoissonInt}) lends to the interpretation of the dynamics $\Phi_{t,\lambda}$ as describing a free particle (no field potential) with Poisson timed collisions occurring at rate $\frac{1}{\lambda}$ with gas particles with statistical outcomes governed by $\mathcal{P}_{\lambda}$.    The limiting dynamics $\Phi_{t,\lambda}^{\diamond}$ can be constructed as  a Dyson series using the group $F_{t}^{\diamond}$ (defined below) and perturbative part $\lambda(\varphi(\cdot)-\frac{1}{2}\{\varphi(I),\cdot\})$.    However, the complete positivity of the maps $\Phi_{t,\lambda}^{\diamond}$ is not apparent without the analysis in the proof of Lemma~\ref{CPSG}.

The strategy of our proof of~(\ref{MainEqu}) is based on showing that $\Phi_{t,\lambda}$ has second-order in $\lambda$ error for an integral equation that the dynamics $\Phi_{t,\lambda}^{\diamond}$ solve.  
The basic idea is that a process close to solving an integral equation is close to the solution of the integral equation.    We show that the  error $E_{t}^{\prime}(\Phi_{\cdot,\lambda}(G))$ for the process $\Phi_{r,\lambda}(G)$ at time $t$ defined by
\begin{align}\label{DysonIntro2}
E_{t}^{\prime}(\Phi_{\cdot,\lambda}(G))=F_{t}^{\diamond}(G)+\lambda\int_{0}^{t}ds\, F_{t-s}^{\diamond}\big(\varphi(\Phi_{s,\lambda}(G))-\frac{1}{2}\{ \Phi_{s,\lambda}(G), \varphi(I)\} \big) -\Phi_{t,\lambda}(G),
\end{align}
where $F_{t}^{\diamond}$ is the unitary evolution generated by $i[ \frac{1}{2M}(\vec{P}+\lambda M\vec{A})^{2}+V_{1}+\lambda V_{2}, \cdot ]$, is uniformly bounded by a constant multiple of $\lambda^2\|G\|_{wn}$ for all $G\in \Bi(L^{2}(\R^{n}))$ and times $t$ in a finite time interval $[0,T]$.   With $\Phi_{t,\lambda}$ replaced by $\Phi_{t,\lambda}^{\diamond}$ in~(\ref{DysonIntro2}), the error on the left side of the equation becomes zero, and we can apply the elementary inequality~(\ref{Pert}) to bound the norm of the difference between $\Phi_{t,\lambda}(G)$ and $\Phi_{t,\lambda}^{\diamond}(G)$.  

  A solution to a differential equation can be written as the solution to different integral equations and we choose~(\ref{DysonIntro2}) because the perturbative part $\varphi(\cdot)-\frac{1}{2}\{\varphi(I),\cdot\}$ is a bounded map.  However,   using the bound~(\ref{Previous}), it is more natural to bound the error $E_{t}(\Phi_{\cdot,\lambda}(G))$ in
\begin{multline}\label{DysonIntro}
E_{t}(\Phi_{\cdot,\lambda}(G))=F_{t}(G)+\int_{0}^{t}ds\, F_{t-s}\big( i[V_{1}+\lambda V_{2}+\frac{\lambda}{2}\{\vec{P},  \vec{A}\}+\frac{\lambda^{2}}{2}M\vec{A}^{2} , \Phi_{s,\lambda}(G) ]\\+\lambda \varphi(\Phi_{s,\lambda}(G))-\frac{\lambda}{2}\{ \varphi(I),\Phi_{s,\lambda}(G)\} \big) -\Phi_{t,\lambda}(G).
\end{multline}
   Since the perturbative part of the integral equation~(\ref{DysonIntro}) involves unbounded terms (the momentum operators $(\vec{P})_{j}$ are unbounded), we can not apply Proposition~\ref{Pert} here.   Thus we use Proposition~\ref{PertProp2} to relate the error $E_{t}(\Phi_{\cdot,\lambda}(G))$ to the error $E_{t}^{\prime}(\Phi_{\cdot,\lambda}(G))$.    A basic necessity of our technique is that we will need to verify that the semigroup $\Phi_{t,\lambda}$ is minimally well behaved with respect to the weighted norm $\|\cdot\|_{wn}$ in the sense that there exists an $ \alpha\geq 0$ such that
\begin{align}\label{Trip}
  \|\Phi_{t,\lambda}(G)\|_{wn}\leq e^{\alpha\, t}\|G\|_{wn},
  \end{align}
for all $G\in \Bi(L^{2}(\R^{n}))$, $t,\lambda>0$.

    In Section~\ref{DysonGen}, we develop general conditions to guarantee that a semigroup obeys an inequality of form~(\ref{Trip}) for some weighted operator norm.  Section~\ref{DysonSpec} checks these conditions specifically for $\Phi_{t,\lambda}$ and $\|\cdot\|_{wn}$ and gives a bound of the type~(\ref{Trip}) which will, importantly, be independent of $\lambda$.   Finally, Section~\ref{SectMain} contains Theorem~\ref{Main}, the proof of which gives the details for the inequality~(\ref{MainEqu}).

\section{Bounding weighted norms of evolved observables }\label{DysonGen}
  In this section, we develop conditions for attaining inequalities of the form 
\begin{align}\label{Minimal}
\|\Phi_{t}(G)\|_{\mathbf{Y}}\leq e^{\alpha\,t}\|G\|_{\mathbf{Y}},
\end{align}
for $\alpha \geq 0$,where $\Phi_{t}:\Bi(\Hi)\rightarrow \Bi(\Hi)$ is a semigroup satisfying  an integral equation
\begin{align}\label{BasicIntEqn}
\Phi_{t}=F_{t}+\int_{0}^{t}ds\, F_{t-s}\Psi \Phi_{s},
\end{align}
 $F_{s},\Psi:\Bi(\Hi)\rightarrow \Bi(\Hi)$ for a complex Hilbert space $\Hi$, $F_{s}$ forms a group (with a possibly unbounded generator),  $\Psi$ is bounded, and $\|\cdot\|_{\mathbf{Y}}$ is an weighted operator norm determined by  densely defined closed operators $Y_{1,j}$, $Y_{2,j}$, $j=1,\cdots, n$ as
$$\|G\|_{\mathbf{Y} }=\|G\|+\sum_{j=1}^{n} \|Y_{1,j}GY_{2,j}^{*}\|.$$
Since the adjoints $Y_{1,j}^{*}$ and $Y_{2,j}^{*}$ are densely defined, the operator $Y_{1,j}GY_{2,j}^{*}$ should be interpreted as determined by a densely defined quadratic form $F(\psi_{1};\psi_{2})= \langle Y_{1,j}^{*}\psi_{1}|G Y_{2,j}^{*} \psi_{2}\rangle$ when $F$ is bounded.   Clearly $\|G\|_{\mathbf{Y}}$ may be infinite for some $G\in \Bi(\Hi)$, so we look a the subspace $\mathcal{A}\subset \Bi(\Hi)$ with finite norm.    The inequality~(\ref{Minimal}) just shows that the dynamics $\Phi_{t}$ is minimally well behaved with respect to the weighted norm $\|\cdot\|_{\mathbf{Y}}$.

 Our technique for bounding $\|\Phi_{t}(G)\|_{\mathbf{Y}}$ relies on the assumption that maps $\Psi_{Y,j,s}$ of the form
\begin{align}\label{Commy}
\Psi_{Y,j}(G)= \Psi(Y_{1,j}GY_{2,j}^{*})-Y_{1,j}\Psi(G)Y_{2,j}^{*}
\end{align}
are bounded.   In practice, it may be useful to shift this condition to certain ``dilated" operators related to $Y_{1,j}$ and $Y_{2,j}$.   For instance if $\Hi=L^{2}(\R^{n})$ and $Y_{1,j}=Y_{2,j}=|\vec{P}|$, then it becomes possible to look at~(\ref{Commy}) with $Y_{1,j}$ replaced by $\vec{P}$ and $Y_{2,j}^{*}$ replaced by $(\vec{P})^{*}$.    The dilations are natural, since we have intertwining relations such as $\tau_{k}\vec{P}=(\vec{P}-\vec{k})\tau_{k}$  for $\tau_{k}=e^{i\vec{k}\vec{X}}$, which we use in the analysis of Section~\ref{DysonSpec}.   In this situation, $\Psi_{Y,j}$ would be a  map from $\Bi(\Hi)$ to $\Bi(\Hi\otimes \C^{n})$, and $\Psi$ is interpreted to act component-wise.     We formalize the concept of a  ``dilation" of an operator in the following definition.  
\begin{definition}
Let $A$ be a closed operator with a dense domain in a Hilbert space $\Hi$.    A dilation of $A$ is a closed operator
$$
\hat{A}: \textup{D}(A)\subset \Hi\rightarrow \Hi\otimes \Hi_{dil},
$$
where $\Hi_{dil}$ is a Hilbert space and $\hat{A}$ satisfies 
 $$\hat{A}^{*}\hat{A}=A^{*}A.$$

\end{definition}

The following lemma shows how dilating an operator in certain expressions leaves the norm unchanged.  
\begin{lemma}\label{AdjDil}
Let $A,B$ be closed operators on a Hilbert space $\Hi$ with dilations $\hat{A},\hat{B}$  corresponding  to dilation spaces $\Hi_{dil}^{(1)}$ and   $\Hi_{dil}^{(2)}$ respectively.  For any $G\in \Bi(\Hi)$,
$$
\|AG B^{*}\|_{\Bi(\Hi) }=\|\hat{A}G\hat{B}^{*}\|_{\Bi(\Hi\otimes \Hi_{dil}^{(1)}, \Hi\otimes \Hi_{dil}^{(2)}  )}.
$$
where the norms can be infinite. 
\end{lemma}

A linear map $\Psi: \Bi(\Hi)\rightarrow \Bi(\Hi)$ is said to be {\em dilation bounded} if
$$
 \| \Psi\otimes I_{\Hi_{aux}}\| =c <\infty
 $$
for an infinite-dimensional Hilbert space $\Hi_{aux}$.  Clearly for such a map, $ \| \Psi\otimes I_{\Hi_{aux}}\| \leq c$ when $\Hi_{aux}$ is finite-dimensional.   In particular, a bounded completely positive map is dilation bounded.

\begin{proposition}\label{SinHitUnBnd2}
Let $\Phi_{t}$ be a semigroup determined by~(\ref{BasicIntEqn}) where  $\Psi$ has a dilation bound $b$, and  $F_{t}$ forms a group.  Assume that $F_{t}$ is an isometry and  that $\Phi_{t}$ is an weak contraction with respect to the operator norm.    Let $\|\cdot \|_{\mathbf{Y} }$ be a generalized operator norm and $\hat{Y}_{1,j},\hat{Y}_{2,j}$ be dilations of $Y_{1,j}, Y_{2,j}$ with dilation spaces $\Hi_{1,dil}$, $\Hi_{2,dil}$.    Suppose also that:

\begin{enumerate}

\item
there exists some constant $a\in \R^{+}$ such that
\begin{align}\label{NotMainCond}
\| F_{t}(G)\|_{\mathbf{Y}}\leq e^{at}\|G\|_{\mathbf{Y}},
\end{align}

\item and defining the commutator $\Psi_{\hat{Y}, j}(G)=\hat{Y}_{1,j}\Psi(G)\hat{Y}_{2,j}^{*}-\Psi(\hat{Y}_{1,j}G \hat{Y}_{2,}^{*})$ for $j=1,\cdots n$, there exists a $c>0$ s.t. for all $j$ and any $G$
\begin{align}\label{MainCond}
\|\Psi_{\hat{Y}, j}(G)\|\leq c \|G\|_{\mathbf{Y}}.
\end{align}

\end{enumerate}
For $\alpha=a+nc$, we have the following inequality:
$$
\| \Phi_{t}(G)\|_{\mathbf{Y}}\leq e^{\alpha\, t}\|G\|_{\mathbf{Y} }.
$$

\end{proposition}

\begin{proof}
Recall the Dyson series identity
\begin{align}\label{Here}
\hat{Y}_{1,j}\Phi_{t}(G)\hat{Y}_{2,j}^{*}=\Phi_{t}(\hat{Y}_{1,j}G \hat{Y}_{2,j}^{*})+\int_{0}^{t}ds\,\Phi_{t-s}(F_{t-s}^{*}\Psi_{\hat{Y}_{1,j},\hat{Y}_{2,j}}F_{t-s})\Phi_{s}(G).
\end{align}
Taking the operator norm of both sides of~(\ref{Here}), and using the fact that $\Phi_{t}$ is contraction w.r.t. the $\infty$-norm and that $F_{t}$ is an isometry, we get the inequality
\begin{multline*}
\|\hat{Y}_{1,j}\Phi_{t}(G) (\hat{Y}_{2,j})^{*}\| \leq \| \hat{Y}_{1,j}F_{t}(G)(\hat{Y}_{2,j})^{*}\|+\int_{0}^{t}ds\, \|\varphi_{\hat{Y}_{1,j},\hat{Y}_{2,j}}F_{t-s}\Phi_{s}(G) \|
\\ \leq \|\hat{Y}_{1,j}F_{t}(G)\hat{Y}_{2,j}^{*}\| +c\int_{0}^{t}ds\,  e^{(t-s)a}\|\Phi_{s}(G)\|_{\mathbf{Y}}.
\end{multline*}
Applying Lemma~\ref{AdjDil} we get
\begin{align*}
\|Y_{1,j}\Phi_{t}(G)Y_{2,j}^{*}\| \leq \|Y_{1,j}F_{t}(G) Y_{2,j}^{*}\|  +c\int_{0}^{t}ds\,  e^{(t-s)a}\|\Phi_{s}(G)\|_{\mathbf{Y}}.
\end{align*}
 Now if we sum over $j$, add $\|\Phi_{t}(G)\|$ to the left and $\|G\|$ to the right (since $\Phi_{t}$ is a contraction), and apply~(\ref{NotMainCond}), then we arrive at
\begin{align*}
\|\Phi_{t}( G)\|_{\mathbf{Y}}\leq \| G \|_{\mathbf{Y}}e^{at} +nc\int_{0}^{t}ds\,  e^{a(t-s)}\|\Phi_{s}(G)\|_{\mathbf{Y}}.
\end{align*}
 Finally, applying Gronwall's inequality, we conclude
$$
\|\Phi_{t}(G)\|_{\mathbf{Y} }\leq e^{(a+nc)t}\|G\|_{\mathbf{Y}}.
$$

\end{proof}

\section{Bounding $\|\Phi_{t,\lambda}(G)\|_{wn}$ }\label{DysonSpec}

The goal of this section is to apply Proposition~\ref{SinHitUnBnd2} to $\Phi_{t,\lambda}$ and $\|\cdot \|_{wn}$ and to get a bound of $\|\Phi_{t,\lambda}(G)\|_{wn}$ independent of $\lambda$.   $\Phi_{t,\lambda}$ satisfies the integral equation    
$$\Phi_{t,\lambda}(G)=F_{t}(G)+\int_{0}^{t}ds\, F_{t-s}\Psi_{\lambda}\Phi_{s,\lambda}(G) \text{, where }$$$$\Psi_{\lambda}(G)=\frac{1}{\lambda}\big( \Tr_{2}[(I\otimes \rho)\Sca^{*}_{\lambda}(G\otimes I)\Sca_{\lambda}]-G\big) \text{ and } F_{t}(G)= e^{\frac{it}{2M}[\vec{P}^{2},\cdot]}(G).$$
  In our setup for an application of~(\ref{SinHitUnBnd2}), the weights $|\vec{X}|$ and $|\vec{P}|^{r}$ in our norm $\|\cdot\|_{wn}$ are dilated as $\vec{X}:L^{2}(\R^{n})\rightarrow L^{2}(\R^{n})\otimes \C^{n}$ and $\vec{P}^{\otimes^{r}}:L^{2}(\R^{n})\rightarrow L^{2}(\R^{n})\otimes \C^{nr}$ respectively, which are convenient for verifying condition~(\ref{MainCond}).    To check the condition~(\ref{MainCond}) we need  to bound commutation  maps    
\begin{itemize}
\item
$
\Psi_{\textbf{1},\lambda }(G)=\vec{P}^{\otimes^{r}}\Psi_{\lambda}(G)(\vec{P}^{\otimes^{s}})^{t}-\Psi_{\lambda}(\vec{P}^{\otimes^{r}}G(\vec{P}^{\otimes^s})^{t} ),
$
\item
$\Psi_{\textbf{2},\lambda}(G)=\vec{X}\Psi_{\lambda}(G)-\Psi_{\lambda}(\vec{X}G)$,
\item
$
\Psi_{\textbf{3},\lambda}^{(S)}(G)=(S\vec{P})^{t}\vec{X}\Psi_{\lambda}(G)-\Psi_{\lambda}((S\vec{P})^{t}\vec{X}G),
$
\end{itemize}
where $S\in M_{n}(\R)$ (and  we only really need the case $S=|i\rangle\langle j|$).   For $\Psi_{\textbf{j},\lambda}$, $j=2,3$ analogous maps are defined with transposed vectors of operators multiplying from the right however the analysis for these is the same.  

The proof of the following proposition will require notation and results from~\cite{Clark}.    In particular, we make use of the forms derived for $\Tr_{2}[(I\otimes \rho)\Asca_{\lambda}^{*}]$ and $\Tr_{2}[(I\otimes \rho)\Asca^{*}_{\lambda}(G\otimes I)\Asca]$ in~\cite[Prop. 2.2]{Clark} and~\cite[Prop. 2.3]{Clark} respectively for $\Asca_{\lambda}=\Sca_{\lambda}-I$.     Controlling the integrals of operators in these derived expressions requires  the use of~\cite[Prop. 3.1]{Clark} and~\cite[Prop. 3.3]{Clark}.  

Define the weighted trace norm
\begin{multline}
\|\rho \|_{wtn}=\|\rho\|_{1} +\sum_{\epsilon}\sum_{1\leq i,j \leq n}\| |\vec{P}|^{n-2+\epsilon }[X_{i},[X_{j},\rho]]\|_{1}\\+\sum_{\epsilon}\sum_{j=1}^{n}\||\vec{P}|^{n-2+\epsilon}X_{j}\rho X_{j}|\vec{P}|^{n-2+\epsilon}\|_{1}+\||\vec{P}|^{2(n-2)}\rho |\vec{P}|^{2(n-2)} \|_{1},
\end{multline}
where the sums in $\epsilon$ are over $\{0,1\}$ and $\{-1,0,1\}$ in the one- and three-dimensional cases respectively.  
  
\begin{proposition}\label{CorMain}
Let $\rho$ satisfy that $\|\rho\|_{wtn}<\infty$, then there exists  $\alpha\geq 0 $ such that for all $0<\lambda$, $G\in \Bi(L^{2}(\R^{n}))$, 
$$\|\Phi_{t,\lambda}(G)\|_{wn}\leq e^{\alpha\,t} \|G\|_{wn}.$$

\end{proposition}
The norm $\|\cdot\|_{wtn}$ is somewhat stronger than what is required for the proof, but it does not help us to improve it since $\|\rho\|_{wtn}$ is used in our application of the result from~\cite{Clark}.  

\begin{proof}
By the construction~(\ref{PoissonInt}), the maps are $\Phi_{t,\lambda}$ completely positive.   Also from~(\ref{PoissonInt}) we see that $\Phi_{t,\lambda}(I)=I$.   Hence the maps $\Phi_{t,\lambda}$ are contractive since $\|\Phi_{t,\lambda}\|=\Phi_{t,\lambda}(I)\|=\|I\|=1$. 

Since $F_{t}$ leaves $\vec{P}$ invariant and $F_{t}(X_{j})=X_{j}+tP_{j}$, using the triangle inequality and that $1+t\leq e^{t}$, it follows that there exists an $a>0$ s.t. $\|F_{t}(G)\|_{wn}\leq e^{at}\|G\|_{wn}$.   Thus condition (1) of Proposition~\ref{SinHitUnBnd2} holds.  Checking the second condition (2) of Proposition~\ref{SinHitUnBnd2} requires many calculations where the basic strategy and techniques are similar in the different cases, so we focus on $\Psi_{\textbf{1},\lambda }(G)$ for $s=0$.
   
 $\Psi_{\lambda,\mathbf{1} }(G)$ can be written as $\Psi_{\mathbf{1},\lambda}(G)=\Psi_{L,\lambda }(G)+\Psi_{C,\lambda}(G),$ where  
 $$\Psi_{L,\lambda}(G)=\big(\vec{P}^{\otimes^{r}}\Tr_{2}[(I\otimes \rho)\Asca_{\lambda}^{*}]-\Tr_{2}[(I\otimes \rho) \Asca_{\lambda}^{*}]\vec{P}^{\otimes^{r} }\big)G \text{, and}  $$
$$\Psi_{C,\lambda}(G)=\vec{P}^{\otimes^{r}}\Tr_{2}[(I\otimes \rho)\Asca_{\lambda}^{*}(G\otimes I) \Asca_{\lambda}]-\Tr_{2}[(I\otimes \rho) \Asca_{\lambda}^{*}(\vec{P}^{\otimes^{r}}G\otimes I) \Asca_{\lambda}]. $$
The above terms $\Psi_{L,\lambda}(G)$ and  $\Psi_{C,\lambda}(G)$ are naturally handled individually.   Using~\cite[Prop. 2.2]{Clark} and~\cite[Prop. 2.3]{Clark} the $\Psi_{L,\lambda}(G)$ and $\Psi_{C,\lambda}(G)$ terms can be written as
\begin{multline*}
\Psi_{L, \lambda}G=\frac{1}{\lambda}\int_{\R^{n}}d\vec{k} \int_{ SO_{n}}d\sigma \,\tau_{\vec{k}}\, \tau_{\sigma \vec{k}}^{*}\, \rho(\ato,  \yto)\\ \scat(|\vec{k}|)\, [ (\vec{P}-(\sigma-I) \vec{k})^{\otimes^{r}}-\vec{P}^{\otimes^{r}}]\, G
\end{multline*}
and
\begin{multline*}
\Psi_{C, \lambda}( G)=\sum_{j} \frac{1}{\lambda}\int_{\R^{n}}d\vec{k} \int_{SO_{n}\times SO_{n}}d\sigma_{1}\,d\sigma_{2}\,U^{*}_{\vec{k}, \sigma_{1},\lambda }\,   \mtoj \\  \hspace{.5cm}  \bar{\scat}_{\lambda}(|\kto | ) \,  [\wtt^{\otimes^{r}} -\vec{P}^{\otimes^{r}}]\, G\, \scat_{\lambda}(| \ktt |)\, \mttj \, U_{\vec{k}, \sigma_{2}\lambda }.
\end{multline*}
respectively, where $\wt=\frac{1+\lambda\sigma}{1+\lambda }\vec{P}-\frac{\sigma-I}{1+\lambda }\vec{k}$.

Starting with the $\Psi_{L,\lambda}(G)$ term, $(\vec{P}-(\sigma-I)\vec{k} )^{\otimes^{r}}-\vec{P}^{\otimes^{r}}$ can be represented as a sum of operators $X_{\epsilon_{1}}\otimes \cdots \otimes X_{\epsilon_{r}}$, where $X_{1}=\vec{k}-\sigma \vec{k}$ and $X_{2}=\vec{P}$, and at a least a single $\epsilon_{j}$ is $1$.
\begin{multline}\label{Yore}
\|\Psi_{L,\lambda}G\|  \leq \sum_{\epsilon_{1},\cdots \epsilon_{r}}\|\int_{\R^{n}}d\vec{k}\int_{ SO_{n}}d\sigma\, \tau_{\vec{k}} \, \tau_{\sigma \vec{k}}^{*}\, \rho(\ato, \yto))\\  \frac{1}{\lambda}\scat(|\vec{k}|)\, X_{\epsilon_{1}}\otimes \cdots X_{\epsilon_{m}}\, G \|
\end{multline}
For each  $X_{\epsilon_{1}}\cdots X_{\epsilon_{r}}$  we can unitarily rearrange all $X_{2}=\vec{P}$ terms to the front of the tensor product since this is merely equivalent to a formal rearrangement of the tensors in the product $L^{2}(\R^{n})\otimes (\C^{n})^{\otimes^{r}}$.  Focusing on a single product with $p<r$ components of $\vec{P}$, it can be rearranged:
\begin{multline}
\|\int_{\R^{n}}d\vec{k}\int_{ SO_{n}}d\sigma \,\tau_{\vec{k}}\, \tau_{\sigma \vec{k}}^{*}\, \frac{\big(\ato-(\att+\sigma \vec{k}-\vec{k})\big)^{\otimes^{r-p}} }{(\delta_{n,3}+|\ato|^{n-2})^{-1} }\, \rho(\ato,\att+(\sigma-I)\vec{k} ) \\ E_{1}(\vec{P},\vec{k},1,\lambda )\, \vec{P}^{\otimes^{p}}(I+|\vec{P}|)\,G \|,
\end{multline}
where we have rewritten $-\sigma \vec{k}+\vec{k}=\ato-(\att+\sigma \vec{k}-\vec{k})$ to match the matrix entries of $\rho$.  Now to set things up for an application of~\cite[Prop. 3.1]{Clark}:
\begin{eqnarray*}
\eta &=& (\delta_{n,3}+|\vec{P}|)^{n-2}\, | [\vec{P} ,\cdot]^{r-p}(\rho)|,
\\
n_{\vec{k},\sigma_{1}}&=&E_{1}(\vec{P},\vec{k},1,\lambda )\, \frac{[\vec{P} ,\cdot]^{r-p}(\rho)}{| [\vec{P} ,\cdot]^{r-p}(\rho)|},
\\
q_{\vec{k},\sigma}&=&n_{\vec{k},\sigma_{1}}\, \eta(\ato ,\yto ),
\end{eqnarray*}
where $n_{\vec{k},\sigma_{1}}$ is treated as a map $n_{\vec{k},\sigma_{1},\lambda}: L^{2}(\R^{n})\otimes \C^{np}\rightarrow  L^{2}(\R^{n})\otimes \C^{nr}$ acting on $ \vec{P}^{\otimes^{p}}(1+|\vec{P}|)G$.  Hence~(\ref{Yore}) is bounded by a constant multiple of
$$\| \eta \|_{1}\| |\vec{P}|^{p}(I+|\vec{P}|)G\|.$$

  For $\Phi_{C,\lambda}(G)$, we need to pick a way of rewriting the expression $ (\wto)^{\otimes^{r}} -\vec{P}^{\otimes^{r}}$.  Using the relation $\wt=\vec{P}-\lambda \frac{(\sigma-I)^{2}}{\sigma+\lambda}\vec{P}-\frac{(1+\lambda)(\sigma-I)}{\sigma+\lambda}\vt$, we can write
\begin{align*}
(\wto)^{\otimes^{r}} -\vec{P}^{\otimes^{r}}=\sum_{(\epsilon_{1},\dots \epsilon_{r})\neq (1,\dots,1)}X_{\epsilon_{1}}\otimes \cdots X_{\epsilon_{r}},
\end{align*}
where $X_{1}=\vec{P}$, $X_{2}=-\lambda \frac{(\sigma-I)^{2} }{\sigma+\lambda}\vec{P}$, $X_{3}=\frac{(1+\lambda)(\sigma-I)}{\sigma+\lambda}\vto $, and the sum is over all cases where not all operators in the product are $\vec{P}$.
\begin{multline*}
\|\varphi_{C,\lambda}(G )\| \leq \sum_{(\epsilon_{1},\dots \epsilon_{r})} \|\sum_{j} \frac{1}{\lambda}\int_{\R^{n}}d\vec{k} \int_{SO_{n}\times SO_{n}}d\sigma_{1}\,d\sigma_{2}\, U^{*}_{\vec{k},\sigma_{1},\lambda} \,   \mtoj^{*}\\  \bar{\scat}_{\lambda}(| \kto |)\,  [ X_{\epsilon_{1}}\otimes \cdots X_{\epsilon_{r}} ]\, G\, \scat_{\lambda}(|\kto |)\, \mttj \,U_{\vec{k},\sigma_{2},\lambda}\|.
\end{multline*}
We break the analysis into two cases: $\epsilon_{i}=3$ for some $i$ and $\epsilon_{i}\neq 3$ for all $i$.

For the case in which  $\epsilon_{i}=3$ for some $i$, it follows that the net power of $\vec{P}$ in the product $X_{\epsilon_{1}}\otimes \cdots \otimes X_{\epsilon_{m}}$ is strictly less than $r$.  Formally rearranging the order of the tensor product $L^{2}(\R^{n})\otimes (\C^{n})^{\otimes r}$ so that the $X_{1}$ components come first, $X_{2}$ components second, and $X_{3}$ components last:   $X_{3}^{r-m_{1}-m_{2}}X_{2}^{m_{2}}X_{1}^{m_{1}}$, $r-m_{1}-m_{2}>0$.
Define the operator  $B_{m_{2},m_{1},\lambda} \in \mathcal{B}(L^{2}(\R^{n})\otimes (\C^{n})^{\otimes^{m_{1}+m_{2}}})$, which operates on a simple tensor product $f(z)\otimes v_{1}\otimes \cdots v_{m}$ as
\begin{align*}
 B_{m_{2},m_{1},\lambda}f(z)\otimes v_{1}\otimes \cdots v_{m} =  f(z)\otimes  \frac{(\sigma-I)^{2} }{\sigma+\lambda }v_{1}\otimes \cdots   \frac{(\sigma-I)^{2}}{\sigma+\lambda } v_{m_{2}} \otimes v_{m_{2}+1}\otimes \cdots v_{m_{1}+m_{2}}.
\end{align*}
We can rearrange our expression as:
\begin{multline*}
\|\sum_{j} (-\lambda)^{m_{2}}\int_{\R^{n}}d\vec{k} \int_{SO_{n}\times SO_{n}}d\sigma_{1}\,d\sigma_{2}\, U^{*}_{\vec{k},\sigma_{1},\lambda}\,  (\delta_{n,3}+|\vto|^{n-2})\, X_{3}^{\otimes^{r-m_{1}-m_{2}} } \, \mtoj ^{*}\\ E_{2}(\vec{P},\vec{k},\sigma_{1},1,\lambda) \,  B_{m_{2},m_{1},\lambda}\,  \vec{P}^{m_{1}+m_{2}}\,(1+|\vec{P})\, G\, \scat_{\lambda}(| \ktt | ) \, \mttj \,U_{\vec{k},\sigma_{2},\lambda }\|.
\end{multline*}
Hence we can apply~\cite[Prop. 3.3]{Clark} as
\begin{eqnarray*}
\eta^{(1)}_{j}(\vec{k})&=& (\delta_{n,3}+|\vec{k}|^{n-2})\,  |\vec{k}|^{r-m_{1}-m_{2}}\,  f_{j}(\vec{k}),
\\
n_{j,\vec{k},\sigma_{1}}^{(1)}( \vec{P})&=&  E_{1}(\vec{P},\vec{k},\sigma_{1},1,\lambda)\,   \frac{\vec{k}^{\otimes^{r-m_{1}-m_{2}}}}{|\vec{k}|^{r-m_{1}-m_{2}}} \, B_{m_{2},m_{1},\lambda} \mtoj,
\\
h_{j,\vec{k},\sigma_{1}}(\vec{P})&=&n_{j,\vec{k},\sigma_{1}}^{(1)}\,\eta^{(1)}_{j}(\vto ),
\\
\eta^{(2)}_{j}(\vec{k})&=& f_{j}(\vec{k}),
\\
n_{j,\vec{k},\sigma_{2}}^{(2)}(\vec{P})&=&  \scat_{\lambda}(| \ktt | )\,  \mttj,
\\
g_{j,\vec{k},\sigma_{2}}(\vec{P})&=&n_{j,\vec{k},\sigma_{2}}^{(2)}\,\eta^{(2)}_{j}(\vtt),
\\
\end{eqnarray*}
where the integral is acting on $   \vec{P}^{\otimes^{m_{1}} }(I+|\vec{P}|)G$.   Hence this expression is bounded by some constant times
$$\| (\delta_{n,3}+|\vec{P}|^{n-2})|\vec{P}|^{r-m_{1}-m_{2}}\rho |\vec{P}|^{r-m_{1}-m_{2}}(\delta_{n,3}+|\vec{P}|^{n-2})\|_{1} \|  |\vec{P}|^{m_{1} }(I+|\vec{P}|)G\|.$$

This concludes the case when at least one of the $\epsilon_{j}$ is $3$.  Now we can move on to the somewhat simpler case when $\epsilon_{j}\neq 3$ for all $j$.  Since at least one of the $\epsilon_{j}$ is   $1$, we have a non-zero power  $\lambda^{m_{2}}$ coming from the $X_{2}=-\lambda \frac{(\sigma-I)^{2}}{\sigma+\lambda}\vec{P}$ terms.  The factor $\frac{1}{\lambda}$ in front of the expression is then cancelled, and it has the form

\begin{multline*}
\|\sum_{j} \lambda^{m_{2}-1}\int_{\R^{n}}d\vec{k}  \int_{SO_{n}\times SO_{n}}d\sigma_{1}\,d\sigma_{2}\,U^{*}_{\vec{k},\sigma_{1},\lambda}\,    \mtoj^{*}\bar{\scat}_{\lambda}(|\kto |)\\     B_{m_{2},m_{1},\lambda} \,  \vec{P}^{\otimes^{r}}\,G\, \scat_{\lambda}(|\ktt | )\,\mttj \,U_{\vec{k},\sigma_{2},\lambda }\|.
\end{multline*}

In this case we can apply~\cite[Prop. 3.3]{Clark}  as
\begin{eqnarray*}
\eta^{(1)}_{j}(\vec{k})&=& f_{j}(\vec{k}),
\\
n_{j,\vec{k},\sigma_{1}}^{(1)}(\vec{P})&=&  \lambda^{m_{2}-1}\, \mtoj \, \bar{\scat}_{\lambda}(| \kto|)\, B_{m_{2},m_{1},\lambda},
\\
h_{j,\vec{k},\sigma_{1}}(\vec{P})&=&n_{j,\vec{k},\sigma_{1}}^{(1)}(\vec{P})\,\eta^{(1)}_{j}(\vto),
\\
\eta^{(2)}_{j}(\vec{k})&=& f_{j}(\vec{k}),
\\
n_{j,\vec{k},\sigma_{2}}^{(2)}(\vec{P})&=&   \mtoj \,\scat_{\lambda}(| \kto|),
\\
g_{j,\vec{k},\sigma_{2}}(\vec{P})&=&n_{j,\vec{k},\sigma_{2}}^{(2)}(\vec{P})\,\eta^{(2)}_{j}(\vtt),
\\
\end{eqnarray*}
where the integral of operators is acting on $\vec{P}^{\otimes^{r}}G$, and we get that the above expression is bounded by some constant multiplied by $\|\rho\|_{1}\| \vec{P}^{\otimes^{r}}G\|=\| |\vec{P}|^{r}G\|$.

\end{proof}

\section{The Main Theorem}\label{SectMain}

Before we get to proving the main result we need to state two inequalities similar to~(\ref{Previous}), except that the errors is of first order rather than second order.   Their proofs are simplified versions of the proof of~(\ref{Previous}) so we omit them.

Define the weighted operator norms
\begin{itemize}

\item $ \|G\|_{\mathbf{a}}= \|G\|+\| |\vec{P}|G\|+\|G |\vec{P}|\|$,

\item

$
\| G  \|_{\mathbf{b}}=\|G\|+\| |\vec{P}|^{2}G\|+\|G|\vec{P}|\|,
$

\end{itemize}
and the weighted trace norms
\begin{itemize}

\item
$
\|\rho \|_{\mathbf{a},1}=\|\rho\|_{1}+\| |\vec{X}|\rho |\vec{X}|\|_{1}+\sum_{\epsilon}\sum_{j}\||\vec{P}|^{n-2+\epsilon}[X_{j},\rho] \|_{1}+\| |\vec{P}|^{n-2} \rho |\vec{P}|^{n-2}\|_{1},
$
\item
$\|\rho \|_{\mathbf{b},1}=\|\rho \|_{1}+\sum_{\epsilon}\sum_{i,j}\| |\vec{P}|^{n-2}[P_{i},[X_{j},\rho ]]\|_{1}+ \| |\vec{P}|^{n-1} \rho |\vec{P}|^{n-1}) \|_{1},$\\

\end{itemize}
where the sums in $\epsilon$ are over $\{0\}$ for dimension one and $\{-1,0\}$ for dimension three.   Note that if we represent the integral kernel of $\rho$ in the momentum basis, then $\|\vec{p}|^{r}\rho |\vec{p}|^{r}\|_{1}=\int_{\R^{n}} d\vec{k}\, |\vec{k}|^{2r}\rho(\vec{k},\vec{k})$ and $[X_{j},\rho](\vec{k}_{1},\vec{k}_{2})=\frac{\partial}{\partial h}\rho(\vec{k}_{1}+h,\vec{k}_{2}+h)|_{h=0}$.   Hence the norms $\|\rho\|_{\mathbf{a}}$ and $\|\rho\|_{\mathbf{b}}$ are like weighted Sobolev $1$-norms on $\rho$.    

Define the maps $\Psi_{\lambda}^{\prime}$ and $\Psi_{\lambda}^{\prime,\diamond}$ on $\Bi(L^{2}(\R^{n} ))$
 \begin{itemize}
\item $
\Psi_{\lambda}^{\prime} (G)=\vec{P}\Psi_{\lambda}(G)-\Psi_{\lambda}(\vec{P}G),
$
\item $ \Psi_{\lambda}^{\prime, \diamond}(G)= \vec{P}\Psi_{\lambda}^{\diamond}(G)-\Psi_{\lambda}^{\diamond}(\vec{P}G).$
\end{itemize}

\begin{lemma}\label{HardLem13}  Let $\Phi_{t,\lambda}$ be defined by

There exists a $c>0$ s.t. for all  $\rho$, $G$,and   $0\leq \lambda$ such that
$$
\|\Psi_{\lambda}(G)-\Psi_{\lambda}^{\diamond} (G)\|\leq c\lambda\|\rho \|_{\mathbf{a} ,1} \|G\|_{\mathbf{a}} \text{ and }
\|\Psi_{\lambda}^{\prime}(G)-\Psi_{ \lambda}^{\prime, \diamond}(G)\|\leq c\lambda\|\rho \|_{\mathbf{b} ,1} \|G\|_{\mathbf{b}}.
$$
\end{lemma}

\begin{theorem}[Main Theorem]\label{Main}
Let  $\rho$ satisfy that $\|\rho\|_{wtn},\|\rho\|_{\mathbf{b},1}<\infty$.   Then over any interval $[0,T]$ there exists a constant $\kappa_{T}$ s.t. for all $t\in [0,T]$, $0\leq \lambda $, and  $G\in \Bi(L^{2}(\R^{n}))$,
$$
\|\Phi_{t,\lambda}(G)-\Phi_{t,\lambda}^{\diamond}(G)\|\leq (t\wedge 1)\kappa_{T}\lambda^{2} \|G\|_{wn}.
$$

\end{theorem}

\begin{proof}
The basic strategy of this proof was discussed in the introduction, and now we fill in the details.  By our conditions on $\rho$, it follows from Proposition~\ref{CorMain} that there exists an $\alpha>0$ such that:
\begin{align}\label{BigBnd}
\|\Phi_{t,\lambda}(G)\|_{wn}\leq e^{\alpha\, t }\|G\|_{wn}
\end{align}
for all $0\leq \lambda$, $G\in \Bi(L^{2}(\R^{n}))$.

 $\Phi_{t,\lambda}(G)$ satisfies the last hit integral equation:
$$
\Phi_{t,\lambda}(G)=F_{t}G+\int_{0}^{t}ds\, F_{t-s}\, \Psi_{\lambda}\,\Phi_{s,\lambda}(G).
$$
By adding and subtracting $\int_{0}^{t}ds\, F_{t-s}(\Psi_{\lambda}^{\diamond}+iM\lambda^{2}[\vec{A}^{2},\cdot])(\Phi_{s,\lambda}(G))$ from the above equation and rearranging terms, we can write:
\begin{align}\label{PreError}
F_{t}(G)+\int_{0}^{t}ds\, F_{t-s}(\Psi_{\lambda }^{\diamond}+iM\lambda^{2}[\vec{A}^{2},\cdot ] ) (\Phi_{s,\lambda}(G))-\Phi_{t,\lambda}(G)=E_{t}(\Phi_{r,\lambda}(G)), \text{ where}
\end{align}
\begin{align}\label{Error}
E_{t}(\Phi_{r,\lambda}(G))=-\int_{0}^{t}ds\, F_{t-s}\,(\Psi_{\lambda}-\Psi_{\lambda}^{\diamond})(\Phi_{s,\lambda}(G))+\lambda^{2}\int_{0}^{t}ds\, F_{t-s}\,i[\vec{A}^{2},\Phi_{s,\lambda}(G) ].
\end{align}
$E_{t}(\Phi_{\cdot,\lambda}(G))$ is understood as the error for the process $\Phi_{r,\lambda}(G)$ to the equation~(\ref{PreError}).   The process $\Phi_{r,\lambda}^{\diamond}(G)$ solves the left side of~(\ref{PreError}) exactly (i.e. with a zero error).   As discussed in the introduction, we would have liked to bound the difference between $\Phi_{r,\lambda}(G)$ and $\Phi_{r,\lambda}^{\diamond}(G)$  by bounding  the error $E_{t}(\Phi_{\cdot,\lambda}(G))$ so that we can apply Proposition~\ref{Pert}.   However, since $\Psi_{\lambda}^{\diamond}$ contains an unbounded part involving commutations with momentum operators, we cannot apply Proposition~\ref{Pert} so directly.   The main technical details of this proof concern using Proposition~\ref{PertProp2} to relate the error $E_{t}(\Phi_{sr\lambda}(G))$ to the error $E_{t}^{\prime}(\Phi_{\cdot,\lambda}(G))$ determined by the integral equation:
\begin{align*}
E_{t}^{\prime}(\Phi_{\cdot,\lambda}(G))=F_{t}^{\diamond}(G)+\int_{0}^{t}ds\,F_{t-s}^{\prime}\,L(\Phi_{s,\lambda}(G)),
\end{align*}
where $L(G)=\varphi(G)-\frac{1}{2}\{G,\varphi(I)\}$ and $F_{t}^{\diamond}$ is the group generated by $i[ \frac{1}{M}(\vec{P}+\lambda M\vec{A})^{2}+V_{1}+\lambda V_{2},\cdot]$.    The commutators with $V_{1}$ and $V_{2}$ could have been left with the perturbative part of the integral equation, but, by convention, we group them with the evolution part of the equation.

Now we  prepare a launch  of Proposition~\ref{PertProp2}.  Define the maps
\begin{eqnarray*}
\Psi_{\lambda,1}^{\prime}(G)&=&\frac{i}{2M}[\vec{P}^2,G],
\\
\Psi_{\lambda,2}^{\prime} (G)&=&i[V_{1}+ \lambda \frac{1}{2}\{\vec{P},\vec{A}\} +\lambda V_{2}+\frac{M\lambda^{2}}{2}\vec{A}^{2} ,G], \text{ and}
\\
\Psi_{\lambda,3}^{\prime}(G)&=&\lambda (\varphi(G)-\frac{1}{2}\{\varphi(I),G\} ).
\end{eqnarray*}

 By Proposition~\ref{PertProp2}, if $\Psi_{\lambda,i}^{\prime}F_{s-r}\,\Psi_{\lambda,j}^{\prime}\Phi_{\lambda,r}(G)$, $i,j=2,3$ make sense and are uniformly bounded for $s,r\in[0,t]$, then we have the following integral relation between the errors as:
\begin{align}\label{PertExam}
E_{t}^{\prime}(\Phi_{\cdot,\lambda}(G))= E_{t}(\Phi_{\cdot,\lambda}(G))+\int_{0}^{t}ds\,F_{t-s}^{\diamond}\,\Psi_{\lambda,2}^{\diamond}\, E_{s}(\Phi_{\cdot,\lambda}(G)).
\end{align}
However, $\| \Psi_{\lambda,i}^{\prime}\,F_{t-s}\Psi_{\lambda,j}^{\prime}\,\Phi_{s,\lambda}(G)\| \leq c\|\Phi_{s,\lambda}(G)\|_{wn}$ for some $c>0$ since the only unbounded  part in the $\Psi_{i}$ expressions are the powers of $P_{i}$, which get up to second-order in the products $\Psi_{i}F_{t-s}\Psi_{j}$ with $i,j\geq 2$, and  $\|\cdot \|_{wn}$ includes such terms.

Now we begin the analysis of bounding $E_{t}^{\prime}(\Phi_{\cdot,\lambda}(G))$.   The  $E_{t}(\Phi_{\cdot,\lambda}(G))$ term can be bounded as:
\begin{multline}\label{FirstBnd11}
\|E_{t}(\Phi_{\cdot,\lambda}(G))\| \leq \int_{0}^{t}ds\, \|(\Psi_{\lambda}-\Psi_{\lambda}^{\diamond} )(\Phi_{s,\lambda}(G))\|+\lambda^{2} M\int_{0}^{t}ds\, \| [\vec{A}^{2},\Phi_{s,\lambda}(G)]\|   \\ \leq \lambda^2 \int_{0}^{t}ds\, (C_{1} \|\Phi_{s,\lambda}(G)\|_{wn}+C_{2}\|\Phi_{s,\lambda}(G)\|_{wn})\leq \lambda^{2}(t\wedge 1) C^{\prime}_{t}\|G\|_{wn},
\end{multline}
where the second inequality follows from~(\ref{Previous}) and the fact that $\vec{A}^{2}$ is a bounded operator for some constants $C_{1}$, $C_{2}$.  The third inequality comes from~(\ref{BigBnd}) and of course that $\|G\|\leq \|G\|_{wn}$.   It follows that $E_{t}(\Phi_{\cdot,\lambda}(G))$ is uniformly bounded and second-order in $\lambda$ over finite time intervals for some constant $C_{t}^{\prime}$ and  is linearly small in $t$ for times near zero.

The other expression on the right side of~(\ref{FirstBnd11}) has  norm less than
$$
\int_{0}^{t}ds\,\|\Psi_{\lambda, 2}^{\diamond}E_{s}(\Phi_{\cdot,\lambda}(G))\|.
$$
It is not immediately clear how to bound this quantity since $\Psi_{\lambda,2}^{\diamond}$ involves unbounded terms involving the vector of momentum operators.   We will reorganize these terms to a form that we can analyze.
Using c.c.r. $\{\vec{P},\vec{A}\}=2\vec{A}\vec{P}+\sum_{j}A_{j}^{\prime}$ or $\{\vec{P},\vec{A}\}=2\vec{P}\vec{A}+\sum_{j}A_{j}^{\prime}$ where $A_{j}^{\prime}=[P_{j},A_{j}]$.  Using these equalities, we can rewrite $\Psi_{\lambda,2}^{\diamond}$ as the following:
$$
\Psi_{\lambda, 2}^{\diamond}(G)=i[V_{1}+\lambda V_{2}+\frac{1}{2} \lambda^2 \vec{A}^2+\lambda \frac{1}{2}\sum_{j}A_{j}^{\prime},G] +  \frac{1}{2}\lambda \sum_{j=1}^{n}i(A_{j}P_{j}G-GP_{j}A_{j}),
$$
where $V_{1}+\lambda V_{2}+ \frac{1}{2}\lambda^2 \vec{A}^2+\lambda \frac{1}{2}\sum_{j}A_{j}^{\prime}$ is  uniformly bounded in operator norm for bounded $\lambda\geq 0 $.  Define
$$
b=2(\| V_{1}\|+\|V_{2}\|+\| |\vec{A}|^{2}\| +\sum_{j=1}^{n}\|A_{j}^{\prime}\|)\wedge \sup_{j}\|A_{j}\|. $$
Then
\begin{multline*}
 \int_{0}^{t}ds\, \|\Psi_{\lambda, 2}^{\diamond}(E_{t}(\Phi_{\cdot,\lambda}( G)) )\| \leq  b \int_{0}^{t}ds\,\| E_{s}(\Phi_{\cdot,\lambda}(G))\| \\+ \lambda \sum_{j=1}^{n} b\int_{0}^{t}ds\,( \|P_{j}\,E_{s}(\Phi_{\cdot,\lambda}(G))\|+ \|E_{s}(\Phi_{\cdot,\lambda}(G))\,P_{j}\|).
\end{multline*}
We have already analyzed $E_{t}(\Phi_{\cdot,\lambda}(G))$ and the first term is bounded by $C_{t}^{\prime}t\|G\|_{wn}$.   Now let us look at the expression $\|P_{j} E_{t}(\Phi_{\lambda,r}(G))\|$.  It can be bounded by
\begin{multline*}
\|P_{j} E_{t}(\Phi_{\cdot,\lambda}(G))\|\leq \|\vec{P}\, E_{t}(\Phi_{\cdot,\lambda}(G))\|
\\ \leq \int_{0}^{t}ds\, \| \vec{P}\,(\Psi_{\lambda}-\Psi_{\lambda}^{\diamond} )(\Phi_{s,\lambda}(G))\|+\frac{\lambda^{2}}{2}M\int_{0}^{t}ds\, \|P_{j}\, [|\vec{A}|^{2}, \Psi_{s,\lambda}(G)]\|.
\end{multline*}
The second expression on the right can be bounded in operator norm by a constant multiple of $\lambda^{2}t\|G\|_{wn}$, since the $|\vec{A}|^{2}$ operator is bounded and the norm $\|\cdot\|_{wn}$ includes weights by the momentum operators $P_{j}$.   The first expression on the right side takes a little more work.   Writing
$$
\vec{P}\Psi_{\lambda}(\Phi_{t,\lambda}(G))=\Psi_{\lambda}(\vec{P}\Phi_{\lambda,r}(G))+\Psi_{\lambda}^{\prime}(\Phi_{t,\lambda}(G)), \text{ and}
$$
$$
\vec{P}\Psi_{\lambda}^{\diamond}(\Phi_{t,\lambda}(G))=\Psi_{\lambda}^{\diamond}(\vec{P}\Phi_{t,\lambda}(G))+\Psi_{\lambda}^{\prime, \diamond}(\Phi_{t,\lambda}(G)),
$$
by Lemma~\ref{HardLem13} there exists a $C_{3}$ s.t.
\begin{align}
\|(\Psi_{\lambda}-\Psi_{\lambda}^{\diamond})\vec{P} \Phi_{r,\lambda}(G)\|\leq C_{3}\lambda \| \vec{P}\,\Phi_{r,\lambda}(G)\|_{\mathbf{a}}  \leq \lambda C_{3}  \|\Phi_{r,\lambda}(G))\|_{wn}\leq  \lambda C_{3}e^{\alpha\,r}\|G\|_{wn},
\end{align}
and there exists a $C_{4}$ s.t.
\begin{align}
 \| (\Psi_{\lambda}^{\prime}-\Psi_{\lambda}^{\prime, \diamond})\Phi_{r,\lambda}(G)\|\leq C_{4}\lambda \|\Phi_{r,\lambda}(G)\|_{\mathbf{b}} \leq C_{4}\lambda \|\Phi_{r,\lambda}(G)\|_{wn}\leq \lambda C_{4} e^{\alpha\,r} \|G\|_{wn}.
\end{align}
So we have
\begin{align*}
\|P_{j} E_{s}(\Phi_{\cdot,\lambda}(G))\| \leq \int_{0}^{s}dr\, \lambda (C_{3}+C_{4})\|\Phi_{r,\lambda}(G)\|_{wn}.
\end{align*}
By similar reasoning, the terms $\| E_{s}(\Phi_{\cdot,\lambda}(G))K_{j}\|$ have the same bound.   So from Proposition~(\ref{PertExam}), we have found that there is a constant $C_{t}^{\prime \prime}$ such that
$$
\sup_{s\in [0,t]}\| E^{\prime}_{s}(\Phi_{\cdot,\lambda}(G)) \| \leq \lambda^{2}(t\wedge 1) C_{t}^{\prime \prime} \|G\|_{wn}.
$$
Hence $\Phi_{t,\lambda}(G)$ has a small error term for the integral equation with free evolution $F^{\diamond}$ and perturbation $\Psi_{\lambda,3}^{\diamond}$ .   $\Psi_{\lambda,3}^{\diamond}$ is bounded so we can apply~(\ref{Pert}) to get a bound for the distance between $\Phi_{t,\lambda}(G)$ and the solution $\Phi_{t,\lambda}^{\diamond}(G)$ of the integral equation.
\begin{align*}
\sup_{s\in [0,t]} \|\Phi_{s,\lambda}(G)-\Phi_{s,\lambda}^{\diamond}(G)\| \leq e^{t\|\Psi_{\lambda,3}\| }\sup_{s\in [0,t]} \| E_{2}(\Phi_{\cdot,\lambda}(G))  \| \leq  t \kappa_{t} \lambda^{2}\|G\|_{wn},
\end{align*}
where $\kappa_{t}=e^{t\|\Psi_{\lambda,3}\| }C_{t}^{\prime \prime}$.       Thus for any finite interval $[0,T]$, $\|\Phi_{t,\lambda}(G)-\Phi_{t,\lambda}^{\diamond}(G)\|\leq (t\wedge 1) \kappa_{T}\lambda^{2}\|G\|_{wn}$.

\end{proof}

\appendix

\begin{center}
    {\bf \Large APPENDIX \normalsize }
  \end{center}

\section{Inequalities involving errors to integral equations}

Th basic idea of the following proposition is that a process having small error in solving an integral equation is close in norm to the solution of the integral equation.  

\begin{proposition}\label{Pert}
Let $\varphi$, $F_{t} $ be bounded linear maps on $ \Bi(\Hi) $ and $F_{t}$  be  be a semigroup with $\|F_{t}\|=1$.   Let $H(s)$ be a process $H:\R^{+}\rightarrow \Bi(\Hi)$ that is  bounded in norm uniformly on finite intervals.   Define the error
$$
E_{t}(H)=F_{t}H(0)+\int_{0}^{t}ds\, F_{t-s}\varphi(H(s))-H(t),
$$
and let $\tilde{H}(t)$ be the solution of the above with $E_{t}(\tilde{H})=0$.  Then

$$
\sup_{s\in [0,t]} \| H(s)-\tilde{H}(s)\|\leq e^{t \|\varphi \| }\sup_{s\in [0,t] } \|E_{s}(H)\|.
$$

\end{proposition}

\begin{proof}
 By the iterative definition   $H^{n+1}(t)= F_{t}(H^{n}(0))+\int_{0}^{t}ds\, F_{t-s} H^{n}(s)$ with $H^{1}(t)= H(t)$, $H^{n+1}\rightarrow \tilde{H}$ uniformly over finite intervals.  By a telescoping series occurring on the left side of the following equation we have:
\begin{multline*}
F_{t}(E_{0}(H))+\sum_{n=1}\int_{0\leq t_{1} \leq \cdots \leq t_{n} \leq t} dt_{1}\cdots dt_{n}\, F_{t-t_{n}} \varphi \cdots \varphi F_{t_{1}}(E_{t_{1}}(H)) \\ =F_{t}(H(0))+\sum_{n=1}\int_{0\leq t_{1} \leq \cdots \leq t_{n} \leq t}dt_{1}\cdots dt_{n}\, F_{t-t_{n}} \varphi \cdots \varphi F_{t_{1}}(H(0))-H(t)=\tilde{H}(t)-H(t).
\end{multline*}
  Taking the  norm of both sides of the above equation, then
$$
 e^{t\| \varphi\| }\sup_{0\leq s \leq t} \|E_{s}(H)\| \geq \|\tilde{H}(t)-H(t)\|,
 $$
which implies the result.

\end{proof}

Solutions to two different integral equations can represent solutions to the same differential equation.           However, if a process $H(t)$ is not a solution to these integral equations, then it will have different errors for the different integral equations.   The following proposition gives an identity relating these errors.

\begin{proposition}[Integral equation error identities]\label{PertProp2}

Let $\Psi_{1}$,$\Psi_{2}$, and $\Psi_{3}$ be linear maps acting on dense  subspaces of  $\mathcal{B}(\mathcal{H})$.  Also let $F_{1}(t)$, $F_{2}(t)$ be semigroups with $\|F_{i}(t)\| \leq 1$ generated by $\Psi_{1}$ and $\Psi_{1}+\Psi_{2}$, respectively.  Finally, let $(G(s))\in \mathcal{B}(\mathcal{H})$ be a process satisfying:
\begin{enumerate}
\item
$G(t)\in \bigcap_{i\geq 2}D(\Psi_{i}F_{s})\cap \bigcap_{i,j\geq 2} D(\Psi_{i}F_{s}\Psi_{j})$ for all $t$ and $s$, and
\item
$G(t)$,$\Psi_{i}(G(t))$, and $\Psi_{i}F_{t-s}\Psi_{j}(G(s))$ for $i,j\geq 2$ are uniformly bounded in operator norm for $t,s$ in finite time intervals.
\end{enumerate}

Define errors $E^{1}_{t}(G)$, $E^{2}_{t}(G)$ for the process $G(s)$ at time $s$ as:
\begin{eqnarray*}
E^{(1)}_{t}(G)&=&F^{(1)}_{t}G_{0}+\int_{0}^{t}ds\, F^{(1)}_{t-s}(\Psi_{2}+\Psi_{3})G(s)-G(t),\\
E^{(2)}_{t}(G)&=&F^{(2)}_{t}G_{0}+\int_{0}^{t}ds\, F^{(2)}_{t-s}\Psi_{3}G(s)-G(t).
\end{eqnarray*}
Then we have the integral identity:
\begin{align}\label{Firs}
E^{(2)}_{t}(G)=E^{(1)}_{t}(G)+\int_{0}^{t}ds\, F^{(2)}_{t-s}\Psi_{2}E^{(1)}_{s}(G),
\end{align}
and inversely
\begin{align}\label{Seco}
E^{(1)}_{t}(G)=E^{(2)}_{t}(G)-\int_{0}^{t}ds\, F^{(1)}_{t-s}\Psi_{2}E^{(2)}_{s}(G).
\end{align}

\end{proposition}

\begin{proof}
~(\ref{Firs}) follows from a little algebra involving two applications of the identity:
\begin{align}\label{thir}
F^{(2)}_{t}(G)=F^{(1)}_{t}(G)+\int_{0}^{t}ds\,F^{(2)}_{t-s}\Psi_{2}F^{(1)}_{s}(G),
\end{align}
and~(\ref{Seco}) is similar.

\end{proof}

\section{Expressions for $V_{1}$, $V_{2}$, $\vec{A}$, and $\varphi $}\label{Form}

The forms for $V_{1}$, $V_{2}$ , $\vec{A}$, and $\varphi $ can be written in terms of an integral kernel for $\rho$ as the following:

\begin{eqnarray}
V_{1}&=& \mathit{c}_{\mathbf{n}}\int_{\R^{+}}dk \,|k|^{-1}\int_{|\vec{v}_{1}|=|\vec{v}_{2}|=k }\, d\vec{v}_{1} d\vec{v}_{2}\,  \rho(\vec{v}_{1}, \vec{v}_{2} )\, e^{i\vec{X}(\vec{v}_{1}-\vec{v}_{2})},\\
V_{2}&=& \mathit{c}_{\mathbf{n}}\int_{\R^{+}}dk\, |k|^{-1} \int_{|\vec{v}_{1}|=|\vec{v}_{2}|=k }d\vec{v}_{1} d\vec{v}_{2}\,   (\vec{v}_{1}+\vec{v}_{2})\nabla_{T}\rho(\vec{v}_{1},\vec{v}_{2}) \, e^{i\vec{X}(\vec{v}_{1}-\vec{v}_{2})},\\
\vec{A}&=& \mathit{c}_{\mathbf{n} }\int_{\R^{+}}dk\, |k|^{-1} \int_{|\vec{v}_{1}|=|\vec{v}_{2}|=k }d\vec{v}_{1} d\vec{v}_{2}\,  \nabla_{T}\rho(\vec{k},\vec{v} )\, e^{i\vec{X}(\vec{v}-\vec{k})},\\
\varphi(G)&=& \mathit{c}_{\mathbf{n}}^{2}\int_{\R^{n}}d\vec{k}\,|\vec{k}|^{-2}\int_{|\vec{v}_{1}|=|\vec{v}_{2}|=|\vec{k}| } d\vec{v}_{1}d\vec{v}_{2}\,   \rho(\vec{v}_{1}, \vec{v}_{2} )\,e^{i\vec{X}(-\vec{v}_{1}+\vec{k})}Ge^{-i\vec{X}(-\vec{v}_{2}+\vec{k})},
\end{eqnarray}
where the constant $\mathit{c}_{n}$ depends on the dimension $n=1,3$ and can be can found in~\cite{Clark}.  The expressions are correct for dimension $n=2$ also except that there is an extra term missing in the expression for $V_{2}$.  In dimension $n=1$, the integrals over surfaces $|\vec{v}_{1}|=|\vec{v}_{2}|=k$ are replaced by sums.  A Kraus decomposition for $\varphi$ can be found, for instance, by taking the singular value decomposition $\rho=\sum_{j}\lambda_{j}|f_{j}\rangle \langle f_{j}|$ where the $f_{j}$'s are orthonormal.  The $m_{j,\vec{k}}$'s operator of~(\ref{Mult}) can then be taken as      
\begin{align}\label{General}
m_{j,\vec{k}}=\sqrt{\lambda_{j}}\mathit{c}_{\mathbf{n}}\int_{|\vec{v}|=|\vec{k}|}d\vec{v}\,f_{j}(\vec{v})\,e^{i\vec{X}(-\vec{v}+\vec{k})}.
\end{align}

\section{Existence and uniqueness of dynamical semigroups}

Treating the equations~(\ref{PoissonProc}) and~(\ref{LimitEQN}) rigorously is not difficult, since the unbounded portion of the generators  are merely the kinetic part involving $\vec{P}$.   In particular, the assumption that $\|\rho\|_{wtn}<\infty$, which we have used to bound the error between $\Phi_{t,\lambda}$ and $\Phi_{t,\lambda}^{\diamond}$, implies that $V_{1}$, $V_{2}$, $\vec{A}$, and $\varphi(I)$   are bounded. 

\begin{lemma}\label{CPSG}
Let $D=\{\psi \in L^{2}(\R^{n}) | \| |\vec{P}|^{2}\psi \|_{2}<\infty \}$ and define
\begin{multline}
\mathcal{L}(\psi_{1};G;\psi_{2})=\big\langle ( i\frac{1}{2M} \vec{P}^{2}-\frac{1}{\lambda}I) \psi_{1}|G \psi_{2}\big\rangle\\+ \big\langle \psi_{1}|G (-i\frac{1}{2M}\vec{P}^{2}+\frac{1}{\lambda}I) \psi_{2}\big\rangle+\big\langle \psi_{1}| \Tr_{2}[(I\otimes \rho)\Sca_{\lambda}^{*}(G\otimes I)\Sca_{\lambda}] \psi_{2}\big\rangle, and
\end{multline}
\begin{multline}
\mathcal{L}^{\diamond}(\psi_{1};G;\psi_{2})=\big\langle \big(i \frac{1}{2M}(\vec{P}+\lambda M\vec{A})^{2}+iV_{1}+i\lambda V_{2}-\frac{\lambda}{2}\varphi(I)\big) \psi_{1}|G \psi_{2}\big\rangle\\+ \big\langle \psi_{1}|G \big( -i\frac{1}{2M}(\vec{P}+\lambda M\vec{A})^{2}-iV_{1}-i\lambda V_{2}-\frac{\lambda}{2}\varphi(I)\big) \psi_{2}\big\rangle+\lambda \big\langle \psi_{1}| \varphi(G)  \psi_{2}\big\rangle.
\end{multline}
There exist unique, conservative, semigroups of maps $\Phi_{t,\lambda}$ and $\Phi_{t,\lambda}^{\diamond}$ such that 
\begin{align*}
\frac{d}{dt}\langle \psi_{1}|\Phi_{t,\lambda}(G)\psi_{2}\rangle=\mathcal{L}(\psi_{1};\Phi_{t,\lambda}(G);\psi_{2})\text{, and } \frac{d}{dt}\langle \psi_{1}|\Phi_{t,\lambda}^{\diamond}(G)\psi_{2}\rangle=\mathcal{L}^{\diamond}(\psi_{1};\Phi_{t,\lambda}^{\diamond}(G);\psi_{2}).
\end{align*}
\end{lemma}
\begin{proof}
The arguments for $\Phi_{t,\lambda}$ and $\Phi_{t,\lambda}^{\diamond}$ are similar, so we will just prove the result for $\Phi_{t,\lambda}^{\diamond}$, since it has slightly more detail.   First we argue that in the domain $D$, the operator
$$L= i\big( \frac{1}{2M}(\vec{P}+\lambda M\vec{A})^{2}+V_{1}+\lambda V_{2}\big)-\frac{\lambda}{2}\varphi(I),$$
is maximal accretive.   By~\cite[Lem. B.1]{Clark} the operators $V_{1}$, $V_{2}$, $\vec{A}$, $\varphi(I)$, and $[P_{j},A_{j}]$ for $j=1,\cdots, n$  are bounded.  The only unbounded terms are $B=\frac{1}{2M}\vec{P}^{2}+\{\vec{A},\vec{P}\}$, which is contained only in the imaginary part of $L$, so we just need show that $B$ is self-adjoint with domain $D$.   However, $\{\vec{A},\vec{P}\}$ is a relatively bounded perturbation of $\frac{1}{2M}\vec{P}^{2}$.    By rewriting  $\vec{P}\vec{A}=\sum_{j}[P_{j}, A_{j}]+\vec{A}\vec{P}$, we have that 
$$\|\{\vec{A},\vec{P}\}\psi \|_{1}\leq  \|\sum_{j}[P_{j}, A_{j}]\| \|\psi\|_{2}+\| |\vec{A}|^{2}\|^{2} \| |\vec{P}|\psi\|_{2}$$  
 for $\psi\in D$.  Applying basic inequalities  $\| |\vec{P}|\psi\|_{2}\leq \frac{1}{\sqrt{2}\epsilon}\|\phi\|_{2}+\frac{\epsilon}{\sqrt{2}}\| \vec{P}^{2}\psi\|_{2}$.   So the coefficient in front of $\|\frac{1}{2M}\vec{P}^{2}\psi\|_{2}$ in  a relative bound inequality can be made  $<1$ and by~\cite{Kato}, $B$ is a self-adjoint operator.  

The operators $\partial_{j}A_{j}$, $V_{1}$, $V_{2}$, $\vec{A}$, and $\varphi(I)$ are all bounded when $\|\rho\|_{wtn}<\infty$.    By~\cite{Holevo2}, it follows that a process $G_{t}\in \Bi(L^{2}(\R^{n}))$ solving
\begin{align}\label{FormGen1} \frac{d}{dt}\langle \psi_{1}|G_{t}\psi_{2}\rangle=\mathcal{L}^{\diamond}(\psi_{1};G_{t};\phi_{2})\end{align}
is equivalent to solving 
$$\langle \psi_{1}| G_{t} \psi_{2}\rangle = \langle e^{-tL}\psi_{1}|G_{t} e^{-tL}\psi_{2}\rangle+\int_{0}^{t}ds\, \langle e^{-(t-s)L}\psi_{1}|\varphi(G_{t}) e^{-(t-s)L}\psi_{2}\rangle.$$
 However, since $e^{-tL}$ and $\varphi$ are bounded maps, the solution $G_{t}$ to the later can be constructed as through a Dyson series in $W_{t}(\cdot)=e^{-tL^{*}}(\cdot)e^{-tL}$ and $\varphi$ operating on $G_{0}=G$.  Call $G_{t}=\Gamma_{t,\lambda}^{\diamond}(G)$, then the complete positivity of the maps $\Gamma_{t,\lambda}^{\diamond}$ follows from the complete positivity of $W_{t}$ and $\varphi$.   

The last thing to show is that $\Phi_{t,\lambda}^{\diamond}$ is conservative.   However, a Dyson series expansion in $F_{t}^{\diamond}$ and $\varphi(\cdot)-\frac{1}{2}\{\varphi(I),\cdot \}$ solving the integral equation~(\ref{DysonIntro2}) (with zero error term) will also be a solution to~(\ref{FormGen1}) and hence be equal to $\Phi_{t,\lambda}^{\diamond}(G)$.   Finally, since $F_{t}^{\diamond}(I)=I$ and $\varphi(I)-\frac{1}{2}\{\varphi(I), I\}=0$, it follows that $\Phi_{t,\lambda}^{\diamond}(I)=I$ .

\end{proof}
\section*{Acknowledgments}
For this work I benefit from the Belgian Interuniversity Attraction Poles Programme P6/02.   I would like to thank  Bruno Nachtergaele for suggestions at the early stages of this manuscript.     Partial financial support has come from Graduate Student Research (GSR) fellowships  funded by the National Science Foundation (NSF  \# DMS-0303316 and DMS-0605342).

\end{document}